\documentclass{article}

\usepackage[preprint]{neurips_2026}
\usepackage{amsthm}
\usepackage{amsmath}
\usepackage{enumitem}
\usepackage{graphicx}
\usepackage{booktabs}
\usepackage{multirow}
\usepackage{ulem} %razi added for \sout

\usepackage{amssymb}   % math (Delta, bold, etc.)
\usepackage{tcolorbox}          % for the prompt box
\usepackage{xcolor}             % colors used in tcolorbox
\usepackage{hyperref}           % for \label and \ref

\usepackage[utf8]{inputenc} % allow utf-8 input
\usepackage[T1]{fontenc}    % use 8-bit T1 fonts
\usepackage{hyperref}       % hyperlinks
\usepackage{url}            % simple URL typesetting
\usepackage{booktabs}       % professional-quality tables
\usepackage{amsfonts}       % blackboard math symbols
\usepackage{nicefrac}       % compact symbols for 1/2, etc.
\usepackage{microtype}      % microtypography
\usepackage{xcolor}         % colors

\newtheorem{theorem}{Theorem}
\newtheorem{proposition}{Proposition}

\newtheorem{definition}{Definition}

\title{From Talking Words to Sharing Thoughts:\\ Scalable Multi-LLM Aggregation via Structured Message Passing}
\author{%
  Niloufar Mehrabi\\
  School of Computing\\
  Clemson University, SC, USA\\
  \texttt{nmehrab@clemson.edu} \\
  \And
  Sayed Pedram Haeri Boroujeni \\
  School of Computing \\
  Clemson University, SC, USA\\
  \texttt{shaerib@clemson.edu} \\
  \AND
  Abolfazl Razi \\
  School of Computing\\
  Clemson University, SC, USA \\
  \texttt{arazi@clemson.edu} \\
}

\begin{document}

\maketitle

\begin{abstract}
The emergence of specialized, domain-tuned Large Language Models (LLMs) has demonstrated that smaller models can achieve expert-level performance in specific tasks, while struggling in out-of-domain settings.  Current ensemble methods to combine their complementary expertise primarily rely on iterative re-prompting or cross-model refinement. These approaches suffer from high computational costs and latency because they require repeated LLM inference calls. Furthermore, naive aggregation often leads to anchor corruption, in which noise propagated from weaker models degrades the performance of the most accurate expert. To address these challenges, we propose a framework that integrates model predictions at the semantic layer using a bipartite factor graph. In this architecture, individual LLMs are represented as variable nodes, while a set of check nodes assess their consistency based on diverse epistemic criteria. We develop a message-passing protocol inspired by error-recovery systems to resolve disagreements iteratively. Furthermore, we introduce an asymmetric damping mechanism that protects high-reliability anchor nodes from being overridden by the ensemble majority. Unlike existing methods, our approach operates on output distributions and requires no additional LLM calls during the refinement phase. Evaluating on four benchmarks, including MMLU, MMLU-Pro, GPQA, and MedMCQA, our method demonstrates a 97\% reduction in token usage and up to a 6$\times$ decrease in API calls, reducing inference time from several minutes to mere milliseconds while consistently outperforming leading multi-agent baselines. These results suggest that graph-based belief propagation is a robust, high-speed, and scalable alternative to the current multi-agent LLM systems. The full pipeline and code will be made public.

\end{abstract}

% ═══════════════════════════════════════════════════════════════
\section{Introduction}
% ═══════════════════════════════════════════════════════════════

% The growing diversity of Large Language Models (LLMs) presents a fundamental opportunity for ensemble methods: \arr{using colon ":" is not very common in scientific texts and more often generated by LLMs. For example, here you can simply use ", since" or " because"} no single model excels across all domains, but different models exhibit complementary strengths~\cite{jiang2023llm}. This is particularly evident among distilled and domain-specific models, which often exhibit complementary strengths. % A model fine-tuned on biomedical literature may exhibit reduced performance on formal reasoning tasks, while a model trained on scientific corpora may be less effective in legal or humanities domains. Even among broadly capable models, consistent cross-domain superiority remains elusive. \arr{The opening is very good but can you give specific example like changing "A model fine-tuned on biomedical literature may exhibit reduced performance on formal reasoning tasks" to "Model XXX which is fine-tuned on biomedical literature (achieving x\% accuracy), exhibits poor performance on formal reasoning tasks (Accuracy: X\%)" to make it more convincing and you can quote/find such results from previous work} 

The growing diversity of Large Language Models (LLMs) presents a fundamental opportunity for ensemble methods, since no single model excels across all domains and different models exhibit complementary strengths~\cite{jiang2023llm}. This is particularly evident among distilled and domain-specific models~\cite{boroujeni2026don}. 
A model fine-tuned on specialized literature may achieve strong in-domain performance while remaining less reliable outside its specialization. For example, Med-PaLM 2 achieves 86.5\% accuracy on MedQA, a USMLE-style medical benchmark, through medical domain adaptation and ensemble refinement~\cite{singhal2025toward}. However, this does not imply uniform superiority across broader reasoning domains, and prior analyses of MMLU show substantial variation across subjects such as formal logic and professional law~\cite{gema2025we}. This specialization--generalization gap motivates robust ensemble fusion across complementary models. This complementarity raises a natural question of can we systematically combine multiple LLMs to exceed the accuracy of any single model? Recent multi-agent frameworks~\cite{yun2026graph, wang2406mixture} have demonstrated that combining multiple LLMs can indeed improve accuracy over individual models. However, these approaches predominantly rely on inter-agent communication at the input layer, where models read, critique, and revise each other's full-text responses across multiple generative rounds. This incurs a token usage of $\mathcal{O}(LNd)$, where $L$ is the number of communication rounds, $N$ is the number of agents, and $d$ is the token length. Consequently, the system latency is dominated by $\mathcal{O}(LN)$ sequential inference calls. For instance, Mixture-of-Agents (MOA)~\cite{wang2406mixture} can require up to 19 API calls and 56k tokens per question, a roughly $50\times$ increase in token overhead compared to a single-LLM counterpart. It is noteworthy that this overhead is not inherent to ensemble integration, but arises from performing inter-model collaboration at the natural-language representation layer.
% It is noteworthy that this overhead is not inherent to ensemble integration, but arises from operating at an inefficient representational level \ar{representation layer [use more accurate terms for technical concepts, even if it sounds repetitive or over-used.]}. 
Once an LLM converts its prediction into natural language, its uncertainty is compressed into a discrete token sequence, requiring other models to infer the underlying belief state from unstructured text. To avoid this bottleneck, we perform integration directly at the output logit layer, where each model's softmax distribution provides a compact representation of its epistemic state, preserving confidence margins and secondary candidate preferences. As an initial rigorous setting, we focus on Multiple-Choice Question Answering (MCQA), where the softmax over $K$ candidates makes the fusion problem well-defined and directly measurable.

\begin{figure}[t]
    \centering
    \includegraphics[width=\textwidth]{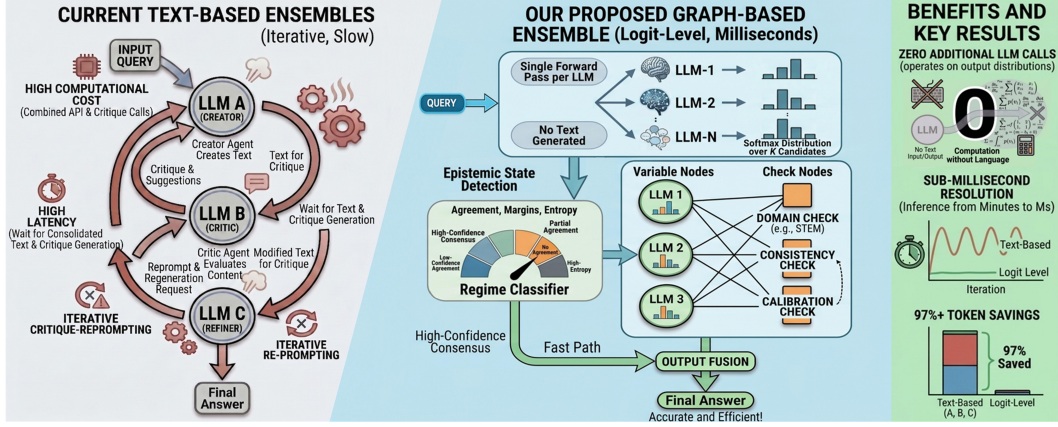}
    \caption{\textbf{Overview of the proposed logit-level multi-LLM ensemble framework.}
    \textbf{Left:} Existing text-based multi-agent systems rely on iterative natural-language communication, requiring repeated LLM calls and causing high token usage, latency, and noisy reasoning propagation.
    \textbf{Middle:} Our framework replaces text exchange with semantic-layer collaboration at the output logit level. Each LLM performs one forward pass to produce a softmax distribution over $K$ candidates. Consensus cases follow fast-path fusion, while disagreement cases activate a factor graph where variable nodes represent model beliefs and check nodes enforce calibration, domain, and agreement constraints through adaptive asymmetric belief propagation.
    \textbf{Right:} This design eliminates additional LLM calls during refinement, enabling sub-millisecond aggregation, lower latency, and over $97\%$ token reduction while matching or outperforming text-based multi-agent baselines.
    % \arr{This figure and caption can be made more informative. The caption does not reflect the image. You can say something like "Left: current XXX methods, where output of each LLM is [augmented with/added to/...?] query..., Middle: proposed method.... Also, the figure is fancy but can be made more informative, like on the left, you can add some legends like Q, IP, and OP to present Query, Input Prompt, and Output Prompt more clearly. The same for middle, reduce the use of fancy icons (for instance, use the same icon for all LLMs on the middle top, and no need to repeat softmax distribution over K candidate three times), and use more relevant and informative text/shapes. Softmas is a typo, I guess. 2- In the middle, the output bar chart of each LLM would be better to match the corresponding bar chart of the variable node. } \nm{Done! We used different icons for LLM models because we wanted to show they are not same models (for example, Qwen, Llama,...)! }
    }
    \label{fig:Overview}
\end{figure}

Specifically, we present a factor-graph framework for logit-level LLM ensemble aggregation. Our method is inspired by sum-product message passing over bipartite graphs used for error recovery of Low-Density Parity-Check (LDPC) in modern communication systems~\cite{7529226}, 
where structured iteration reconciles noisy and partially unreliable observations into a globally consistent solution. This approach enables iterative belief updates to refine LLM outputs, thereby replacing text exchange across consecutive LLM calls. In our framework, each LLM is represented as a variable node holding a probability distribution over answer candidates, and a set of check nodes enforcing soft consistency constraints between models with respect to properties such as per-model calibration, domain expertise, and inter-model agreement patterns, forming a bipartite graph. Before iterating, our method classifies the ensemble's joint epistemic state, based on inter-model agreement, confidence margins, and a domain group mapped from the question's subject, into one of two general regimes, consensus or disagreement. Questions falling into the consensus regime are resolved immediately without message passing, reserving computation for genuinely contested cases.

% \arr{better to skip sparse here, because in your fig all check nodes are connected to all variable nodes making a dence and indeed fully connected graph. if you highlight sparsity, you should say later "to benefit from accelerated performance of BP over sparse bipartite graphs with fewer short loops, we impose sparsity through random edge dropping or ...." } 

% \arr{You may mentioned later but we need to clarify the vision, are we cross-checkinf variable nodes from different perspectives to identify the level of agreement using different metrics or do we cross-check them from different perspectives. 2-Is it a binary consensus vs disagreement or a soft measure gauging the level of agreement and using it to optimally and iteratively refine LLM opinions? I lean toward the latter. }   \arr{Never mind, I see later that you frame the entire work based on two regimes, consensus vs disagreement, so stick with your method and disregard my second point to avoid substantial changes!}

A central challenge in heterogeneous ensembles is \emph{anchor corruption}, where naive aggregation degrades the strongest model by incorporating noise from weaker ones~\cite{zeng2026harnessing}. Our framework addresses this through adaptive asymmetric damping, where a model's receptiveness to incoming messages is governed by a dynamically derived trust ratio. By integrating historical reliability with per-question confidence, this mechanism naturally shields high-reliability models. 
% Unlike static weighting schemes, this approach allows the damping factors to emerge directly from the epistemic context of the ensemble, ensuring structural stability without the need for manual, heuristic-based tuning of individual model weights. 
The final answer is produced by iteratively integrating the models' refined belief distributions, with weights reflecting domain-conditional accuracy, pairwise relational reliability, and individual calibration quality. Critically, our aggregation procedure requires exactly one forward pass per model, totaling $N$ total passes with zero inter-model communication, followed by iterative belief updates. This stands in direct contrast to the $\mathcal{O}(LN)$ calls required by existing frameworks, effectively reducing the iterative generation overhead from $\mathcal{O}(LNd)$ to $\mathcal{O}(1)$ additional tokens beyond the initial inference.

Through comprehensive evaluation on four benchmarks (MMLU, MMLU-Pro, GPQA, and MedMCQA), we demonstrate that our logit-level architecture yields over a $97\%$ reduction in token usage, up to a $3$--$6\times$ decrease in API calls, and reduces inference time from several minutes to mere milliseconds compared to text-based systems,  while consistently matching or outperforming leading multi-agent baselines across these diverse datasets. 

Our primary contributions are as follows:
\begin{itemize}
    \item \textbf{Logit-Level Multi-LLM Collaboration:} We reframe the multi-LLM combination as a distribution fusion problem operating directly at the output layer, replacing the $\mathcal{O}(LNd)$ text-exchange bottleneck with message passing on a bipartite factor graph that requires exactly $N$ forward passes and completes in sub-millisecond time on CPU.
    
    \item \textbf{Epistemic Regime Detection and Routing:} We introduce a principled classification of the ensemble's joint belief state into consensus and disagreement regimes, allowing the framework to bypass message passing for uncontested questions and concentrate computation on cases with genuine inter-model disagreement.

    \item \textbf{Anchor Corruption Analysis:} We formally identify anchor corruption as a fundamental failure mode of naive aggregation in heterogeneous ensembles, and address it through an adaptive damping scheme that protects the most reliable model from distortion by weaker ensemble members without any hyperparameter tuning.

\end{itemize}

% ═══════════════════════════════════════════════════════════════
\section{Related Works}
% ═══════════════════════════════════════════════════════════════

% \arr{no need for titles, because we have only two paragraphs}
% \paragraph{\sout{Multi-LLM Frameworks.}}
Coordinating multiple LLMs to jointly produce better answers has recently gained momentum, with the majority of methods relying on natural-language interaction among models~\cite{du2024improving}. Du et al.~\cite{du2024improving} showed that multiple LLM instances can propose and debate answers over several rounds to improve factuality and reasoning, describing this approach as a society of minds. Liang et al.~\cite{liang2024encouraging} extended this direction with a Multi-Agent Debate framework designed to address the Degeneration-of-Thought problem in self-reflection, while Chan et al.~\cite{chan2023chateval} introduced ChatEval, where role-specialized LLM agents discuss and evaluate generated responses. More recently, Mixture-of-Agents (MoA)~\cite{wang2406mixture} organized agents into multiple generative layers, where each agent conditions its response on the outputs of agents in the previous layer. A recent study revisiting MoA~\cite{li2025rethinking} found that aggregating outputs from a single strong model can sometimes outperform multi-model mixing, raising questions about whether model diversity at the text level is always beneficial. Additionally, SELENE~\cite{verma2026selene} introduced selective debate initiation with evidence-weighted judging, reducing token usage by nearly 50\% while maintaining or improving accuracy and calibration. In these frameworks, inter-agent communication is primarily mediated through natural-language responses and successive LLM recalls. This incurs an $\mathcal{O}(LNd)$ token cost and requires models to reconstruct each other's belief states from unstructured text, a process that can introduce interpretation errors and compounding latency.

% \arr{Good point; although our method has lower token use, the downside is that it is just post-processing of original results. I think we also should think of ways to orchestrate LLM collaboration by making a BP algorithm in a way that triggers/invokes LLMs to rethink and generate better results [for the next paper]!!!}

% \arr{make some connections to the previous, like an alternative way is ....}
An alternative direction avoids explicit debate by aggregating multiple model outputs after inference, without natural-language communication between models. For instance, Self-consistency~\cite{wang2022self} samples multiple reasoning paths to select the most consistent answer, demonstrating that simple aggregation improves accuracy without additional training. LLM-Blender~\cite{jiang2023llm} extends this idea to heterogeneous model outputs through pairwise ranking followed by generative fusion, though it trades inference efficiency for output quality. Beyond voting and generative fusion, routing-based methods such as Smoothie~\cite{guha2024smoothie} perform label-free selection by modeling the relationship between candidate LLM outputs and latent true outputs. FusionRoute~\cite{xiong2026token} explores token-level multi-LLM collaboration, using a lightweight router to select an expert at each decoding step and add complementary logits for next-token refinement. LatentMAS~\cite{zou2025latent} moves closer to semantic-level collaboration by enabling LLM agents to generate latent thoughts using last-layer hidden states and exchange information through shared latent working memory stored in KV caches, rather than relying on natural-language messages. However, LatentMAS performs collaboration inside the LLM inference process through latent thought generation and KV-cache transfer, whereas our method operates after a single forward pass on final softmax distributions. Thus, our approach requires no latent generation, KV-cache exchange, or additional model-internal interaction, and performs all collaboration via lightweight message passing over candidate-level beliefs.

% \arr{this method sounds good [if make experts exchange semantic info, and then it is close to ours], and we need to clarify why our method is better!!!}. \nm{Done!}

% Our method departs from both families by operating directly on softmax distributions through message passing on a sparse factor graph. \arr{before this sentence, you need to say something like, [the common spirit/the essence] of these methods is . [be precise, if applicable, you can say they all operate by re-invoking LLMs over and over]....then you can say, In contrast, we do this and that} It requires exactly $N$ forward passes and sub-millisecond CPU aggregation while explicitly modeling inter-model agreement, domain-conditional reliability, and calibration quality as first-class signals.

% ═══════════════════════════════════════════════════════════════
\section{Problem Definition}
% ═══════════════════════════════════════════════════════════════

\begin{definition}[Multi-LLM Ensemble for MCQA]
\label{def:ensemble}
Given $N$ language models $\{M_1, \ldots, M_N\}$ and a question $q$ with $K$ candidate answers $\mathcal{A} = \{a_1, \ldots, a_K\}$, let $a^* \in \mathcal{A}$ denote the ground truth. Each model $M_i$ produces a belief distribution $b_i \in \Delta^{K-1}$ via the softmax of its output logits $\ell_i$:
\begin{equation}
\label{eq:belief}b_i(a_k) = \frac{\exp(\ell_i(a_k))}{\sum_{j=1}^K \exp(\ell_i(a_j))}, \quad k = 1, \ldots, K
\end{equation}
Our objective is to derive an aggregated belief $b^* = \mathcal{G}(b_1, \dots, b_N)$ that maximizes $P(\arg\max_{a_k} b^*(a_k) = a^*)$ while adhering to strict efficiency constraints.
\end{definition}

\paragraph{Inference Efficiency.}
To ensure scalability, the operator $\mathcal{G}$ requires exactly $N$ total LLM forward passes (one per model) with zero inter-agent textual communication. All post-extraction fusion is computable in $\mathcal{O}(1)$ time relative to LLM inference.

% \paragraph{The Anchor Corruption Limit.} 
A central challenge in heterogeneous ensembles is the risk that weaker models may degrade the performance of a superior "anchor" model. We formalize this via the following proposition:

\begin{proposition}[Anchor Corruption]
\label{prop:anchor_corruption}
Let $M_{i^*}$ be an anchor model with accuracy $p_{i^*} > 0.5$, and let $M_1, M_2$ be two weaker models with equal accuracy $p_w < p_{i^*}$. Under a majority voting scheme with conditionally independent errors, the ensemble accuracy $p_{\mathrm{ens}}$ is strictly less than $p_{i^*}$ if:
\begin{equation}
\label{eq:corruption_condition}
p_w < \frac{p_{i^*} - \sqrt{p_{i^*}(1-p_{i^*})}}{2p_{i^*} - 1}
\end{equation}
The detailed proof is provided in Appendix~\ref{app:anchor_proof}. For clarity, a summary of notation is provided in Appendix~\ref{app:notation}.
\end{proposition}

% \begin{remark}[Design Implications]
% This result demonstrates that naive ensembling can be counterproductive, particularly in heterogeneous settings. It necessitates two core features of the [our method's name] architecture: (1) \textbf{asymmetric damping} to protect high-reliability anchor nodes from being overridden by noise, and (2) \textbf{sparse connectivity} to ensure that weaker models only influence the ensemble when they provide high-confidence, complementary signals.
% \end{remark}

% ═════════════════════════════════════════════════════════════════════════
\section{Proposed Method:} \label{sec:method}
% ═════════════════════════════════════════════════════════════════════════

\subsection{Architecture Overview}

Figure~\ref{fig:Overview} provides an overview of the proposed framework, \textit{\textbf{From Talking Words to Sharing Thoughts}}, 
% a semantic-layer fusion approach that performs logit-level belief propagation over a factor graph for efficient and robust multi-LLM integration. We introduce [our method's name], 
a dynamic inference graph designed to aggregate multiple language models without the severe latency of iterative text generation. By evaluating the structural domain of the query alongside the joint uncertainty profile of the ensemble, the framework conditionally routes execution through a structured message-passing architecture. Crucially, all internal computations operate solely on the initial belief distributions $\{\mathbf{b}_i\}$ obtained from the original forward passes. This design completely eliminates the need for subsequent autoregressive generation steps, successfully resolving complex inter-model disagreements while maintaining strict sub-millisecond inference efficiency.

% ─────────────────────────────────────────────────────────────────────────
\subsection{Epistemic State Detection}\label{sec:layer1}
% ─────────────────────────────────────────────────────────────────────────
Traditional ensembles treat model outputs as isolated votes, collapsing complex probability distributions into scalar confidences. We instead diagnose the geometry of disagreement before aggregation begins. By computing a joint epistemic signature that captures entropy, inter-model margins, distributional shapes, and pairwise divergences, we explicitly map the ensemble's uncertainty profile.

\begin{definition}[Epistemic Signature]
For beliefs $\{\mathbf{b}_1, \ldots, \mathbf{b}_N\}$, the signature $\sigma(q)$ is the tuple:
\begin{equation}
\sigma(q) = \bigl(
    \{H_i\}_{i=1}^N,\;
    \{c_i\}_{i=1}^N,\;
    \{m_i\}_{i=1}^N,\;
    \{\tau_i\}_{i=1}^N,\;
    \{D_{ij}\}_{i<j},\;
    s
\bigr)
\end{equation}
where $H_i = -\sum_k b_{ik} \log b_{ik}$ is the Shannon entropy, $c_i = \max_k b_{ik}$ is the top-1 confidence, and $m_i = b_{i(1)} - b_{i(2)}$ is the decision margin (with $b_{i(k)}$ denoting the $k$-th largest probability mass). Furthermore, $\tau_i \in \{\textsc{peaked}, \textsc{spread}, \textsc{bimodal}\}$ represents the categorical uncertainty geometry (formally derived in Section~\ref{sec:topology}), $D_{ij} = \frac{1}{2}\bigl(D_{\mathrm{KL}}(\mathbf{b}_i \parallel \mathbf{b}_j) + D_{\mathrm{KL}}(\mathbf{b}_j \parallel \mathbf{b}_i)\bigr)$ is the symmetric Kullback-Leibler divergence, and $s = \frac{1}{N} \max_k \sum_{i=1}^N \mathbb{I}(\arg\max \mathbf{b}_i = k)$ is the top-1 agreement.
\end{definition}

This signature relies entirely on the foundational forward passes, allowing it to be evaluated rapidly in $O(N^2 K)$ time, where $K$ is the number of candidate answers, representing a negligible computational overhead. Based on $\sigma(q)$, the query is classified into one of six distinct epistemic regimes, as summarized in Table~\ref{tab:epistemic_regimes}. This regime acts as the central algorithmic router. Structurally unanimous regimes (such as \textsc{Consensus-High}) bypass the factor graph entirely and proceed directly to a fast-path output fusion. Conversely, disjointed regimes trigger the full iterative belief propagation network to resolve the detected conflicts.

\begin{table}[htbp]
\centering
\small
\caption{Classification conditions for the six epistemic regimes based on the joint epistemic signature.}
\label{tab:epistemic_regimes}
\begin{tabular}{@{}ll@{}}
\toprule
\textbf{Regime} & \textbf{Condition} \\
\midrule
\textsc{Consensus-High} & All $N$ agree, all $c_i > \theta_c$ (confidence threshold)\\
\textsc{Consensus-Low}  & All $N$ agree, some $c_i \leq \theta_c$ (confidence threshold) \\
\textsc{Split-Confident}  & $(N{-}1)$ agree, all $c_i > 0.5$ \\
\textsc{Minority-Dissent} & $(N{-}1)$ agree, some $c_i \leq 0.5$ \\
\textsc{Full-Disagreement}& No $(N{-}1)$ majority, moderate entropy \\
\textsc{Confused}         & No $(N{-}1)$ majority, high entropy \\
\bottomrule
\end{tabular}
\end{table}

While the epistemic regime dictates whether the factor graph is instantiated, the relational updates within that graph are strictly conditioned on the query's subject matter. To capture this, the framework simultaneously performs a domain coset classification. Drawing on classical coding theory, where complex signals are grouped into broader structural categories to guide error correction, we define a mapping $\phi: \mathcal{Q} \to \mathcal{D}$ that assigns each query $q$ to a macroscopic domain $d \in \{d_1, \ldots, d_L\}$ (e.g., STEM, Humanities). This domain strictly anchors all subsequent relational calibrations. Formally, for any inter-model prediction pattern $\pi_{ij}$, the conditional probability evaluates as $P(M_i \text{ correct} \mid \pi_{ij}, d) \neq P(M_i \text{ correct} \mid \pi_{ij})$, indicating that model reliability and inter-model dependence are non-stationary across different subjects. In standardized benchmarks, this mapping utilizes existing dataset metadata at zero computational cost, while in zero-shot environments, the domain fingerprint can be dynamically inferred from the joint distribution shapes computed within the epistemic signature.

% ─────────────────────────────────────────────────────────────────────────
\subsection{Bipartite Constraint Topology}
\label{sec:topology}
% ─────────────────────────────────────────────────────────────────────────
For queries exhibiting structural disagreement, we instantiate a dynamic factor graph $\mathcal{G} = (\mathcal{V}, \mathcal{F}, \mathcal{E})$ to execute Belief Propagation (BP). Our architecture draws structural inspiration from classical LDPC decoding~\cite{7529226}. While classical LDPC operates on binary log-likelihood ratios transmitted over noisy channels (see Appendix~\ref{app:Background} for the formal sum-product equations), our framework adapts these concepts to continuous, multi-class probability spaces generated by LLMs.

% \paragraph{Variable Nodes ($\mathcal{V}$).}
Each LLM is represented as a variable node $v_i \in \mathcal{V}$, initialized with the prior belief distribution $\mathbf{b}_i$. 
To strictly enforce the \emph{extrinsic information principle} of LDPC decoding (where a node updates its belief using intrinsic observations alongside incoming messages while excluding its own previous outgoing messages), updates in our graph are computed relative to the initial state. This prevents self-reinforcing echo effects during collaboration.
% \paragraph{Check Nodes ($\mathcal{F}$).}
Check nodes enforce relational and contextual constraints. Inspired by LDPC parity rules, each check node provides an independent measure of consistency. 
The active subset adapts to the specific epistemic regime:

\begin{itemize}[nosep,leftmargin=*]
\item \textbf{Distributional Profiling Node} ($\psi_{\text{prof}}$): Classifies each model's belief shape as \textsc{peaked}, \textsc{bimodal}, or \textsc{spread} based on confidence, margin, and entropy thresholds, then weights messages accordingly: peaked models contribute decisively, bimodal models contribute only their top-2 candidates, and spread models act as weak evidence.

\item \textbf{Calibrated Confidence Node} ($\psi_{\text{conf}}$): Weights messages by the alignment between each model's reported confidence and its historical accuracy at that confidence level.

\item \textbf{Margin Decisiveness Node} ($\psi_{\text{marg}}$): Evaluates the gap between the top-2 candidates ($m_i$), weighted by each model's overall accuracy. High-margin messages from accurate models are amplified as decisive signals, while low-margin messages are treated as tentative.

\item \textbf{Dissent Node} ($\psi_{\text{diss}}$): Active in \textsc{Minority-Dissent} and \textsc{Confused} regimes. Routes messages asymmetrically: majority models hear from dissenters, minority models hear from the majority, with historically accurate minorities receiving boosted weight.

\item \textbf{Domain Expertise Node} ($\psi_{\text{dom}}$): Maps the query to a domain group and weights each model's message by its domain-conditional accuracy, amplifying contributions from models with demonstrated expertise in the relevant subject area.

\item \textbf{Domain-Relational Node} ($\psi_{\text{dom\text{-}rel}}$): Combines domain-specific expertise with pairwise relational accuracy patterns. The message weight for model $j$ toward model $i$ is determined by $(\text{acc}(j,d) \cdot R_{ji}^{(d)})^{\gamma}$, where $R_{ji}^{(d)}$ captures the probability that model $j$ is correct given its agreement or disagreement pattern with model $i$ within domain $d$. This provides the finest-grained inter-model consistency signal in the factor graph.
\end{itemize}

% \paragraph{Edges ($\mathcal{E}$).}
In the default configuration, each variable node is connected to all active check nodes, forming a dense bipartite graph. The active subset of check nodes adapts to the detected epistemic regime: the margin ($\psi_{\text{marg}}$), calibration ($\psi_{\text{conf}}$), domain ($\psi_{\text{dom}}$), and domain-relational ($\psi_{\text{dom-rel}}$) nodes are active in all regimes that invoke BP. The distributional profiling node ($\psi_{\text{prof}}$) activates in \textsc{Split-Confident}, \textsc{Full-Disagreement}, and \textsc{Confused} regimes, while the dissent node ($\psi_{\text{diss}}$) activates only in \textsc{Minority-Dissent} and \textsc{Confused} regimes. This regime-conditional gating deactivates irrelevant check nodes entirely rather than selectively disconnecting individual edges, ensuring that each model receives corrective messages from all applicable constraints at each iteration.

% ─────────────────────────────────────────────────────────────────────────
\subsection{Logit-Level Message Passing and Asymmetric Damping}\label{sec:message_passing}
% ─────────────────────────────────────────────────────────────────────────
\subsubsection{Belief Softening (Channel Calibration)}

A primary obstacle in adapting belief propagation to instruction-tuned LLMs is their tendency toward overconfidence, often yielding near-deterministic distributions. In a standard BP framework, these sharply peaked priors saturate the variable nodes, causing incoming messages to have a negligible impact and effectively inhibiting the exchange of probability mass.

Drawing on the LDPC principle of calibrating channel Log-Likelihood Ratios (LLRs) to prevent premature decoder convergence, we introduce a pre-propagation softening phase. Before BP begins, the initial beliefs are flattened using temperature scaling with $T_{\text{soft}} > 1$:

\begin{equation}
\tilde{b}_{ik} = \frac{(b_{ik})^{1/T_{\text{soft}}}}{\sum_j (b_{ij})^{1/T_{\text{soft}}}}
\end{equation}

This allows the factor graph to iteratively route and accumulate evidence. Once the BP algorithm converges, the resulting distributions are re-sharpened using $T_{\text{sharp}} < 1$ to recover a decisive predictive distribution for the final aggregation stage.

\subsubsection{Adaptive Asymmetric Damping}
\label{sec:damping}

In heterogeneous ensembles, naive belief aggregation may lead to \textit{Anchor Corruption}, wherein the predictive performance of the most accurate model is degraded by conflicting signals from the surrounding ensemble. To mitigate this, we introduce an adaptive damping factor $\alpha_i \in [0, 1]$ designed to dynamically regulate the influence of incoming messages based on the relative expertise and certainty of the nodes. We first map the domain-conditional accuracy ($\text{acc}_i^{(d)}$) of each model $M_i$ to an epistemic evidence score:
\begin{equation}
E_i = \frac{\text{acc}_i^{(d)}}{1 - \text{acc}_i^{(d)}}
\end{equation}

The damping factor $\alpha_i$ is then defined as the complement of the model's normalized contribution to the ensemble's aggregate evidence, scaled by its top-1 confidence $c_i$:
\begin{equation}
\label{eq:posterior_trust}
\alpha_i = 1 - \frac{\bigl( E_i \cdot c_i \bigr)^{\gamma}}{\sum_{j=1}^{N} \bigl( E_j \cdot c_j \bigr)^{\gamma}}
\end{equation}
where $\gamma > 0$ is a sensitivity hyperparameter that governs the sharpness of the trust distribution. This formulation establishes a principled competitive balance. Models that exhibit high evidence and high localized certainty (large $E_i \cdot c_i$) command a larger share of the denominator, driving their specific $\alpha_i$ toward $0$. This renders such nodes structurally resistant to external signals, thereby protecting the anchor's initial belief. Conversely, models with lower relative evidence or high uncertainty receive an $\alpha_i$ closer to $1$, making them highly receptive to corrective messages from the ensemble constraints. The resulting log-domain belief update for each variable node $v_i$ is:
\begin{equation}
\label{eq:bp_update}
\log \mathbf{b}_i^{(t+1)} = \log \mathbf{b}_i^{(0)} + \alpha_i \sum_{f \in \mathcal{F}} \log \mathbf{m}_{f \to i}^{(t)}
\end{equation}
where $\mathbf{m}_{f \to i}^{(t)}$ denotes the message sent from check node $f$ to variable node $v_i$ at iteration $t$. Detailed check-node weighting rules are provided in Appendix~\ref{app:check_weights}. Crucially, by anchoring the update to the initial state $\mathbf{b}_i^{(0)}$ rather than the current state, the framework prevents runaway feedback loops and ensures that the initial model expertise is preserved against systemic noise.

\begin{theorem}[Anchor Stability Bound]
\label{thm:anchor_protection}
Under adaptive asymmetric damping where the anchor model $M_{i^*}$ possesses a damping factor $\alpha_{i^*}$, and given the non-cumulative extrinsic update rule in Eq.~\ref{eq:bp_update}, the anchor's maximum belief deviation at any iteration $T$ is strictly bounded independently of $T$:
\begin{equation}
\label{eq:anchor_bound}
D_{TV}(\mathbf{b}_{i^*}^{(0)}, \mathbf{b}_{i^*}^{(T)})
\leq
\alpha_{i^*} \cdot \sum_{f \in \mathcal{F}(i^*)} \mathcal{E}_f
\end{equation}
where $D_{TV}(\cdot,\cdot)$ denotes the total variation distance between two belief distributions, and $\mathcal{E}_f = \max_k \left| \log \mathbf{m}_{f \to i^*}(a_k) \right|$ is the maximum epistemic evidence of check node $f$. Under the adaptive formulation in Eq.~\ref{eq:posterior_trust}, $\alpha_{i^*} \to 0$ as the anchor's sharpened evidence $(E_{i^*} \cdot c_{i^*})^{\gamma}$ dominates the ensemble total $\sum_j (E_j \cdot c_j)^{\gamma}$, ensuring that the anchor's belief deviation remains tightly constrained regardless of the number of propagation rounds $T$. The detailed proof is provided in Appendix~\ref{app:anchor_stability}.
\end{theorem}

%─────────────────────────────────────────────────────────────────────────
\subsection{Relational Calibration and Output Fusion}\label{sec:fusion}
% ─────────────────────────────────────────────────────────────────────────

The final ensemble prediction is obtained by synthesizing the post-BP beliefs through a weighted combination. Unlike static weighting schemes, this mechanism derives model influence dynamically, conditioned on the observed inter-model prediction patterns and the specific domain context. The final ensemble belief $\hat{\mathbf{b}}$ for each candidate answer $a_k$ is computed through a relationally calibrated fusion of the posterior beliefs $b_i^{\text{out}}$:
\begin{equation}
\label{eq:final_aggregation}
\hat{\mathbf{b}}(a_k) = \frac{\sum_{i=1}^{N} w_i \cdot b_i^{\text{out}}(a_k)}{\sum_{i=1}^{N} w_i}, \quad \hat{y} = \arg\max_{a_k} \hat{\mathbf{b}}(a_k)
\end{equation}
The aggregation weight $w_i$ for model $M_i$ is determined by the joint calibration score $\Phi_i$, scaled by the concentration hyperparameter $\gamma$:
\begin{equation}
\label{eq:relational_weight}
w_i = \Phi_i^{\gamma}, \quad \text{where} \quad \Phi_i = \text{acc}_i^{(d,r)} \cdot \bar{R}_i \cdot \text{cal}_i(c_i) \cdot \text{acc}_i^{(d)}
\end{equation}

The historical expertise and real-time reliability of each model are integrated The historical expertise and real-time reliability of each model are integrated into $\Phi_i$ through several key components, including the model's baseline domain accuracy $\text{acc}_i^{(d)}$ and its specialized historical performance $\text{acc}_i^{(d,r)}$ within specific epistemic regimes. To account for real-time model behavior, the framework incorporates a calibration score $\text{cal}_i(c_i)$ that maps current confidence to expected accuracy, alongside the mean relational weight $\bar{R}_i = \frac{1}{N-1}\sum_{j \neq i} P(M_i \text{ correct} \mid \pi_{ij}, d)$. This final term quantifies the model's credibility by aggregating the strength of agreement or dissent from its neighboring agents within the factor graph.

% ─────────────────────────────────────────────────────────────────────────
\section{Experiment}
\label{sec:Experiment}
% ─────────────────────────────────────────────────────────────────────────
To evaluate the efficacy and efficiency of our proposed framework, we conduct experiments across four high-stakes reasoning benchmarks. Our analysis examines whether semantic-layer fusion improves predictive accuracy, whether adaptive damping mitigates anchor corruption in heterogeneous ensembles, and whether logit-level collaboration reduces token consumption and latency. Compared with traditional text-based multi-agent systems, our logit-level belief propagation maintains strong performance while substantially reducing the computational footprint.

% ─────────────────────────────────────────────────────────────────────────
\subsection{Experimental Setup}
\label{sec:Setup}
% ─────────────────────────────────────────────────────────────────────────

We evaluate our framework across four diverse, reasoning-intensive benchmarks, including MMLU (57 subjects)~\cite{hendrycks2020measuring}, MMLU-Pro (10-choice format)~\cite{wang2024mmlu}, GPQA (graduate-level science)~\cite{rein2023gpqa}, and MedMCQA (specialized medical knowledge)~\cite{pal2022medmcqa}. A comprehensive breakdown of these datasets is provided in Appendix~\ref{app:Dataset_details}. To ensure a fair and rigorous comparison with state-of-the-art graph-based ensembling, we adopt the experimental configuration reported in the GoA framework~\cite{yun2026graph}, utilizing a standardized ensemble of 7-8B parameter models. This expert ecosystem includes Qwen2.5-7B-Instruct (General)~\cite{qwen2024}, Qwen2.5-Coder-7B-Instruct (Code)~\cite{hui2024qwen2}, Mathstral-7B-v0.1 (Math)~\cite{mathstral2024}, Bio-Medical-Llama-3-8B (Biomedical)~\cite{biomedicalllama2024}, finance-Llama3-8B (Finance)~\cite{cheng2024instruction}, and Saul-7B-Instruct-v1 (Legal)~\cite{colombo2024saullm}. Crucially, our framework does not rely on a fixed anchor; instead, the anchor model $M_{i^*}$ is identified adaptively for each query by maximizing a weighted combination of domain-conditional and overall historical accuracy. We compare our approach against state-of-the-art baselines, including single-agent experts and multi-agent generative frameworks such as Debate~\cite{du2024improving}, Self-Consistency (SC)~\cite{wang2022self}, Refine~\cite{madaan2023self}, ReConcile~\cite{chen2024reconcile}, MoA~\cite{wang2406mixture}, Self-MoA~\cite{li2025rethinking}, GOA\_Max~\cite{yun2026graph}, and GOA\_mean~\cite{yun2026graph}. All experiments were conducted using NVIDIA A100 GPUs for LLM inference, while the proposed belief propagation aggregation runs on CPU with an overhead below 0.2 ms per query.

% ─────────────────────────────────────────────────────────────────────────
\subsection{Comparative Performance}
% ─────────────────────────────────────────────────────────────────────────

Primary accuracy results across all benchmarks are reported in Table~\ref{tab:overall_results}. Our method establishes a new state-of-the-art on all four datasets, achieving 79.80\% on MMLU, 55.36\% on MMLU-Pro, 43.64\% on GPQA, and 62.83\% on MedMCQA. These gains are particularly pronounced on high-difficulty benchmarks, where our approach delivers a 10.81\% absolute improvement on GPQA, highlighting its effectiveness in resolving complex reasoning disagreements. These results provide strong empirical evidence that our method mitigates the well-known anchor corruption effect in heterogeneous ensembles. Naive aggregation strategies, such as majority voting, often degrade strong models due to the influence of weaker agents, as reflected by the drop to 73.51\% on MMLU. In contrast, our adaptive asymmetric damping mechanism preserves and selectively amplifies reliable signals, enabling consistent improvements even at the individual model level, as seen in the gains from Anchor model (Initial Belief) to Anchor model (Refined Belief).

\begin{table}[h]
\caption{Accuracy comparison across four reasoning benchmarks. Bold values indicate the best performance in each dataset. Our relational aggregation framework achieves the highest accuracy on all benchmarks, outperforming both single-agent experts and leading text-based multi-agent baselines.}
\label{tab:overall_results}
\centering
\small
\setlength{\tabcolsep}{4pt}
\renewcommand{\arraystretch}{0.80}
\begin{tabular}{llcccc}
\toprule
Category & Method & MMLU & MMLU-Pro & GPQA & MedMCQA \\
\midrule

\multirow{6}{*}{Single-Agent}
& General \cite{qwen2024}                 & 77.61 & 53.90 & 32.83 & 55.22 \\
& Code \cite{hui2024qwen2}                & 68.04 & 42.33 & 33.84 & 45.57 \\
& Math \cite{mathstral2024}               & 63.47 & 37.19 & 30.81 & 45.25 \\
& Biomedical \cite{biomedicalllama2024}   & 46.60 & 27.90 & 25.25 & 47.00 \\
& Finance \cite{cheng2024instruction}     & 54.11 & 25.52 & 28.28 & 42.08 \\
& Legal \cite{colombo2024saullm}          & 55.86 & 27.57 & 30.30 & 41.50 \\

\midrule
\multirow{8}{*}{Multi-Agent}
& Debate \cite{du2024improving}            & 72.53 & 47.05 & 29.29 & 53.05 \\
& SC \cite{wang2022self}                   & 77.97 & 54.12 & 36.36 & 55.70 \\
& Refine \cite{madaan2023self}             & 77.40 & 54.71 & 38.92 & 54.94 \\
& ReConcile \cite{chen2024reconcile}       & 69.61 & 44.19 & 34.34 & 54.60 \\
& MoA \cite{wang2406mixture}               & 75.71 & 53.33 & 32.83 & 54.94 \\
& Self-MoA \cite{li2025rethinking}         & 78.14 & 54.19 & 33.84 & 55.56 \\
& GoA$_{\text{Max}}$ \cite{yun2026graph}   & 79.18 & 54.78 & 39.98 & 60.04 \\
& GoA$_{\text{Mean}}$ \cite{yun2026graph}  & 78.52 & 54.27 & 40.54 & 57.92 \\

\midrule
\multirow{5}{*}{Ours}
& Anchor (Initial Belief) & 77.31 & 52.30 & 40.61 & 59.23 \\
& Anchor (Refined Belief) & 78.43 & 53.21 & 42.14 & 61.12 \\
& Majority (Initial Belief) & 73.51 & 51.02 & 37.26 & 57.47 \\
& Majority (Refined Belief) & 76.22 & 52.37 & 39.41 & 58.87 \\
& \textbf{Relational Aggregation (Ours)} & \textbf{79.80} & \textbf{55.36} & \textbf{43.64} & \textbf{62.83} \\

\bottomrule
\end{tabular}
\end{table}

% ─────────────────────────────────────────────────────────────────────────
\subsection{Efficiency and Scalability}
\label{sec:efficiency}
% ─────────────────────────────────────────────────────────────────────────

A critical advantage of our architecture is the decoupling of consensus from generative latency. As shown in Table~\ref{tab:efficiency}, we evaluate our framework against the leading ensemble baselines, MoA and GoA, across the MMLU-Pro and GPQA benchmarks. While these existing systems require between 11 and 19 sequential API calls and up to 245 seconds to resolve a single query, our method achieves superior accuracy with only 3 parallel API calls and under 0.2 milliseconds of BP overhead. The observed 97\% reduction in token consumption is a direct result of our semantic-layer approach; by eliminating the need for models to iteratively critique and refine each other's reasoning via text, we shift the ensembling bottleneck from GPU-intensive text generation to a negligible, CPU-bound mathematical operation. This effectively removes the multiplicative overhead inherent in multi-turn generative systems. Our analysis shows that accuracy remains robust as the ensemble size grows to $N=10$. While text-based methods can degrade when weaker agents are added, as seen in MoA dropping to 31.82\%, our framework reaches 41.82\% by resolving conflicts at the logit level. This semantic-layer aggregation distills collective signal while avoiding the logical interference of iterative text exchange. Detailed scalability results are provided in Table~\ref{tab:gpqa_10agents} of Appendix~\ref{app:Scalability}.

\begin{table}[h]
\caption{Efficiency comparison on MMLU-Pro and GPQA. Our method achieves the highest accuracy while dramatically reducing computational cost.}
% , requiring only 3 API calls and sub-millisecond belief propagation overhead.} 
\label{tab:efficiency}
\renewcommand{\arraystretch}{0.80}
\centering
\small
\begin{tabular}{llcccc}
\toprule
Dataset & Method & Accuracy & Calls & Tokens(k) & Time(s) \\
\midrule
\multirow{3}{*}{MMLU-Pro}
& MoA \cite{wang2406mixture}     & 53.33 & 19 & 56.05 & 240.26 \\
& GoA-Max \cite{yun2026graph}    & 54.78 & 11 & 19.18 & 100.43 \\
& \textbf{Ours}                  & \textbf{55.36} & \textbf{3} & \textbf{0.357} & \textbf{0.15} \\
\midrule
\multirow{3}{*}{GPQA}
& MoA \cite{wang2406mixture}     & 32.83 & 19 & 56.87 & 245.65 \\
& GoA-Max \cite{yun2026graph}    & 39.98 & 11 & 17.32 & 88.15 \\
& \textbf{Ours}                  & \textbf{43.64} & \textbf{3} & \textbf{0.663} & \textbf{0.20} \\
\bottomrule
\end{tabular}
\end{table}

% Our analysis shows that accuracy remains robust as the ensemble size grows to $N=10$. Traditional text-based frameworks often suffer a performance collapse when adding weaker agents, as seen in MoA dropping to 31.82\%. This occurs because flawed natural language arguments from weaker models interfere with the reasoning of stronger agents. In contrast, our framework reaches 41.82\% by resolving conflicts at the logit level. By operating at the semantic layer, we distill the collective signal and avoid the logical interference typical of iterative text exchange. We provide a detailed breakdown of these scalability results across all baselines in Table~\ref{tab:gpqa_10agents} of the Appendix~\ref{app:Scalability}.

% ─────────────────────────────────────────────────────────────────────────
\subsection{Ablation Study}
% ─────────────────────────────────────────────────────────────────────────
We evaluate the individual contribution of each check node across all benchmarks. As shown in Table~\ref{tab:ablation_multi_dataset}, each relational constraint provides measurable gains over majority-based aggregation, while the integrated All (Ours) configuration achieves the highest overall performance. Among individual components, the margin-based constraint consistently yields the strongest improvements across most benchmarks, while coset-relational and calibration-based constraints are particularly effective in specific domains. Combining them leads to the best performance, confirming that diverse relational signals are essential for improving ensemble accuracy and mitigating anchor corruption.

\begin{table}[h]
\caption{Ablation study across four benchmarks. Individual check nodes improve over majority-based aggregation, while combining all relational constraints achieves the best overall performance.}
\label{tab:ablation_multi_dataset}
\centering
\small
\begin{tabular}{lcccc}
\hline
\textbf{Method} & \textbf{MMLU} & \textbf{MMLU-Pro} & \textbf{GPQA} & \textbf{MedMCQA} \\
\hline
Majority (Initial Belief) & 73.5 & 51 & 37.26 & 57.47 \\
Majority (Refined Belief) & 76.2 & 52.37 & 39.41 & 58.87 \\
\hline
margin only & 79.51 & 54.90 & 43.30 & 62.3 \\
calibration only & 79.17 & 54.75 & 43.10 & 62.21 \\
dissent only & 78.82 & 54.45 & 42.87 & 61.95 \\
agreement only & 79.41 & 54.40 & 42.65 & 61.85 \\
domain only & 78.91 & 54.85 & 42.9 & 62.2 \\
coset\_relational only & 79.24 & 54.95 & 42.75 & 62.05 \\
\hline
\textbf{All (Ours)} & \textbf{79.71} & \textbf{55.50} & \textbf{43.64} & \textbf{62.83} \\
\hline
\end{tabular}
\end{table}

% ─────────────────────────────────────────────────────────────────────────
\section{Conclusion}
\label{sec:Conclusion}
% ─────────────────────────────────────────────────────────────────────────

We presented a semantic-layer fusion framework that resolves the fundamental trade-off between ensemble accuracy and computational overhead. By shifting multi-agent coordination from natural language dialogue to logit-level belief propagation, our method effectively eliminates anchor corruption and achieves a 97\% reduction in token consumption. Experimental results across four reasoning benchmarks demonstrate that our approach not only establishes a new state-of-the-art but also remains robust as ensemble size increases. While our framework demonstrates strong performance and efficiency, its evaluation is restricted to multiple-choice benchmarks, where model outputs can be directly represented as probability distributions over discrete candidates. Extending the framework to open-ended generation or free-form reasoning would require additional mechanisms to align semantic outputs.

\bibliographystyle{unsrt}
\bibliography{main}

@article{yun2026graph,
  title={Graph-of-Agents: A Graph-based Framework for Multi-Agent LLM Collaboration},
  author={Yun, Sukwon and Peng, Jie and Li, Pingzhi and Fan, Wendong and Chen, Jie and Zou, James and Li, Guohao and Chen, Tianlong},
  journal={arXiv preprint arXiv:2604.17148},
  year={2026}
}

@article{wang2406mixture,
  title={Mixture-of-agents enhances large language model capabilities, 2024},
  author={Wang, Junlin and Wang, Jue and Athiwaratkun, Ben and Zhang, Ce and Zou, James},
  journal={URL https://arxiv. org/abs/2406.04692},
  volume={1},
  number={2},
  pages={5}
}

@misc{qwen2024,
  title={{Qwen2.5}: A Party of Foundation Models},
  author={{Qwen Team}},
  year={2024},
  howpublished={\url{https://qwen.ai/blog?id=qwen2.5}},
  note={September 2024}
}

@article{hui2024qwen2,
  title={Qwen2. 5-coder technical report},
  author={Hui, Binyuan and Yang, Jian and Cui, Zeyu and Yang, Jiaxi and Liu, Dayiheng and Zhang, Lei and Liu, Tianyu and Zhang, Jiajun and Yu, Bowen and Lu, Keming and others},
  journal={arXiv preprint arXiv:2409.12186},
  year={2024}
}

@misc{mathstral2024,
  title={{Math$\Sigma$tral}},
  author={{Mistral AI}},
  year={2024},
  howpublished={\url{https://mistral.ai/news/mathstral/}},
  note={July 2024}
}

@misc{biomedicalllama2024,
  title={{Bio-Medical-Llama-3-8B}: A High-Performance Biomedical Language Model},
  author={{ContactDoctor}},
  year={2024},
  howpublished={Online},
  url={https://huggingface.co/ContactDoctor/Bio-Medical-Llama-3-8B}
}

@inproceedings{cheng2024instruction,
  title={Instruction pre-training: Language models are supervised multitask learners},
  author={Cheng, Daixuan and Gu, Yuxian and Huang, Shaohan and Bi, Junyu and Huang, Minlie and Wei, Furu},
  booktitle={Proceedings of the 2024 Conference on Empirical Methods in Natural Language Processing},
  pages={2529--2550},
  year={2024}
}

@article{colombo2024saullm,
  title={Saullm-7b: A pioneering large language model for law},
  author={Colombo, Pierre and Pires, Telmo Pessoa and Boudiaf, Malik and Culver, Dominic and Melo, Rui and Corro, Caio and Martins, Andre FT and Esposito, Fabrizio and Raposo, Vera L{\'u}cia and Morgado, Sofia and others},
  journal={arXiv preprint arXiv:2403.03883},
  year={2024}
}

@inproceedings{du2024improving,
  title={Improving factuality and reasoning in language models through multiagent debate},
  author={Du, Yilun and Li, Shuang and Torralba, Antonio and Tenenbaum, Joshua B and Mordatch, Igor},
  booktitle={Forty-first international conference on machine learning},
  year={2024}
}

@article{wang2022self,
  title={Self-consistency improves chain of thought reasoning in language models},
  author={Wang, Xuezhi and Wei, Jason and Schuurmans, Dale and Le, Quoc and Chi, Ed and Narang, Sharan and Chowdhery, Aakanksha and Zhou, Denny},
  journal={arXiv preprint arXiv:2203.11171},
  year={2022}
}

@article{madaan2023self,
  title={Self-refine: Iterative refinement with self-feedback},
  author={Madaan, Aman and Tandon, Niket and Gupta, Prakhar and Hallinan, Skyler and Gao, Luyu and Wiegreffe, Sarah and Alon, Uri and Dziri, Nouha and Prabhumoye, Shrimai and Yang, Yiming and others},
  journal={Advances in neural information processing systems},
  volume={36},
  pages={46534--46594},
  year={2023}
}

@inproceedings{chen2024reconcile,
  title={Reconcile: Round-table conference improves reasoning via consensus among diverse llms},
  author={Chen, Justin and Saha, Swarnadeep and Bansal, Mohit},
  booktitle={Proceedings of the 62nd Annual Meeting of the Association for Computational Linguistics (Volume 1: Long Papers)},
  pages={7066--7085},
  year={2024}
}

@article{li2025rethinking,
  title={Rethinking mixture-of-agents: Is mixing different large language models beneficial?},
  author={Li, Wenzhe and Lin, Yong and Xia, Mengzhou and Jin, Chi},
  journal={arXiv preprint arXiv:2502.00674},
  year={2025}
}

@inproceedings{jiang2023llm,
  title={Llm-blender: Ensembling large language models with pairwise ranking and generative fusion},
  author={Jiang, Dongfu and Ren, Xiang and Lin, Bill Yuchen},
  booktitle={Proceedings of the 61st Annual Meeting of the Association for Computational Linguistics (Volume 1: Long Papers)},
  pages={14165--14178},
  year={2023}
}

@inproceedings{liu2024dynamic,
  title={A dynamic LLM-powered agent network for task-oriented agent collaboration},
  author={Liu, Zijun and Zhang, Yanzhe and Li, Peng and Liu, Yang and Yang, Diyi},
  booktitle={First Conference on Language Modeling},
  year={2024}
}

@inproceedings{liang2024encouraging,
  title={Encouraging divergent thinking in large language models through multi-agent debate},
  author={Liang, Tian and He, Zhiwei and Jiao, Wenxiang and Wang, Xing and Wang, Yan and Wang, Rui and Yang, Yujiu and Shi, Shuming and Tu, Zhaopeng},
  booktitle={Proceedings of the 2024 conference on empirical methods in natural language processing},
  pages={17889--17904},
  year={2024}
}

@inproceedings{verma2026selene,
  title={SELENE: Selective and Evidence-Weighted LLM Debating for Efficient and Reliable Reasoning},
  author={Verma, Akshay and Gupta, Swapnil and Gupta, Deepak and Sircar, Prateek and Pillai, Siddharth},
  booktitle={Proceedings of the 19th Conference of the European Chapter of the Association for Computational Linguistics (Volume 5: Industry Track)},
  pages={95--104},
  year={2026}
}

@article{chan2023chateval,
  title={Chateval: Towards better llm-based evaluators through multi-agent debate},
  author={Chan, Chi-Min and Chen, Weize and Su, Yusheng and Yu, Jianxuan and Xue, Wei and Zhang, Shanghang and Fu, Jie and Liu, Zhiyuan},
  journal={arXiv preprint arXiv:2308.07201},
  year={2023}
}

@article{guha2024smoothie,
  title={Smoothie: Label free language model routing},
  author={Guha, Neel and Chen, Mayee F and Chow, Trevor and Khare, Ishan S and Re, Christopher},
  journal={Advances in Neural Information Processing Systems},
  volume={37},
  pages={127645--127672},
  year={2024}
}

@article{xiong2026token,
  title={Token-Level LLM Collaboration via FusionRoute},
  author={Xiong, Nuoya and Zhou, Yuhang and Zeng, Hanqing and Chen, Zhaorun and Huang, Furong and Bi, Shuchao and Zhang, Lizhu and Zhao, Zhuokai},
  journal={arXiv preprint arXiv:2601.05106},
  year={2026}
}

@article{zou2025latent,
  title={Latent collaboration in multi-agent systems},
  author={Zou, Jiaru and Yang, Xiyuan and Qiu, Ruizhong and Li, Gaotang and Tieu, Katherine and Lu, Pan and Shen, Ke and Tong, Hanghang and Choi, Yejin and He, Jingrui and others},
  journal={arXiv preprint arXiv:2511.20639},
  year={2025}
}

@article{hendrycks2020measuring,
  title={Measuring massive multitask language understanding},
  author={Hendrycks, Dan and Burns, Collin and Basart, Steven and Zou, Andy and Mazeika, Mantas and Song, Dawn and Steinhardt, Jacob},
  journal={arXiv preprint arXiv:2009.03300},
  year={2020}
}

@article{wang2024mmlu,
  title={Mmlu-pro: A more robust and challenging multi-task language understanding benchmark},
  author={Wang, Yubo and Ma, Xueguang and Zhang, Ge and Ni, Yuansheng and Chandra, Abhranil and Guo, Shiguang and Ren, Weiming and Arulraj, Aaran and He, Xuan and Jiang, Ziyan and others},
  journal={Advances in Neural Information Processing Systems},
  volume={37},
  pages={95266--95290},
  year={2024}
}

@article{rein2023gpqa,
  title={Gpqa: A graduate-level google-proof q\&a benchmark},
  author={Rein, David and Hou, Betty Li and Stickland, Asa Cooper and Petty, Jackson and Pang, Richard Yuanzhe and Dirani, Julien and Michael, Julian and Bowman, Samuel R},
  journal={arXiv preprint arXiv:2311.12022},
  year={2023}
}

@inproceedings{pal2022medmcqa,
  title={Medmcqa: A large-scale multi-subject multi-choice dataset for medical domain question answering},
  author={Pal, Ankit and Umapathi, Logesh Kumar and Sankarasubbu, Malaikannan},
  booktitle={Conference on health, inference, and learning},
  pages={248--260},
  year={2022},
  organization={PMLR}
}

@ARTICLE{7529226,
  author={Durisi, Giuseppe and Koch, Tobias and Popovski, Petar},
  journal={Proceedings of the IEEE}, 
  title={Toward Massive, Ultrareliable, and Low-Latency Wireless Communication With Short Packets}, 
  year={2016},
  volume={104},
  number={9},
  pages={1711-1726},
  doi={10.1109/JPROC.2016.2537298}}

@inproceedings{zeng2026harnessing,
  title={Harnessing consistency for robust test-time LLM ensemble},
  author={Zeng, Zhichen and Yu, Qi and Lin, Xiao and Qiu, Ruizhong and Ning, Xuying and Wei, Tianxin and Yan, Yuchen and He, Jingrui and Tong, Hanghang},
  booktitle={Findings of the Association for Computational Linguistics: EACL 2026},
  pages={3528--3545},
  year={2026}
}

@article{singhal2025toward,
  title={Toward expert-level medical question answering with large language models},
  author={Singhal, Karan and Tu, Tao and Gottweis, Juraj and Sayres, Rory and Wulczyn, Ellery and Amin, Mohamed and Hou, Le and Clark, Kevin and Pfohl, Stephen R and Cole-Lewis, Heather and others},
  journal={Nature medicine},
  volume={31},
  number={3},
  pages={943--950},
  year={2025},
  publisher={Nature Publishing Group US New York}
}

@inproceedings{gema2025we,
  title={Are we done with mmlu?},
  author={Gema, Aryo Pradipta and Leang, Joshua Ong Jun and Hong, Giwon and Devoto, Alessio and Mancino, Alberto Carlo Maria and Saxena, Rohit and He, Xuanli and Zhao, Yu and Du, Xiaotang and Madani, Mohammad Reza Ghasemi and others},
  booktitle={Proceedings of the 2025 Conference of the Nations of the Americas Chapter of the Association for Computational Linguistics: Human Language Technologies (Volume 1: Long Papers)},
  pages={5069--5096},
  year={2025}
}

@article{boroujeni2026don,
  title={Don't Waste Bits! Adaptive KV-Cache Quantization for Lightweight On-Device LLMs},
  author={Boroujeni, Sayed Pedram Haeri and Mehrabi, Niloufar and Woods, Patrick and Hillesheim, Gabriel and Razi, Abolfazl},
  journal={arXiv preprint arXiv:2604.04722},
  year={2026}
}

%%%%%%%%%%%%%%%%%%%%%%%%%%%%%%%%%%%%%%%%%%%%%%%%%%%%%%%%%%%%
\clearpage
\appendix
\section{Appendix}
\label{app:Appendix}

This appendix provides supplementary notation, proofs, theoretical foundations, and extended experimental results. We first summarize the key notation and parameter definitions in Appendix~\ref{app:notation}. We then present the formal derivation for the anchor corruption limit in Appendix~\ref{app:anchor_proof} and the stability analysis for our adaptive asymmetric damping mechanism in Appendix~\ref{app:anchor_stability}, followed by detailed formulations of the message-passing and check-node weighting mechanisms in Appendix~\ref{app:check_weights}. To bridge communication theory and multi-agent systems, Appendix~\ref{app:Background} provides a detailed review of the principles of belief propagation in factor graphs. Technical implementation details, including stratified sampling protocols for our benchmarks and the zero-shot prompt templates, are provided in Appendix~\ref{app:Dataset_details} and Appendix~\ref{app:prompts}, respectively. Extended experimental results, including scalability analysis, domain-specific performance, and entropy convergence dynamics, are provided in Appendix~\ref{app:Extended_analysis}. Finally, a discussion on the environmental sustainability and ethical considerations of our framework is presented in Appendix~\ref{app:Broader_impact}.

% ═══════════════════════════════════════════════════════════════
\section{Notation and Parameter Definitions}
\label{app:notation}
% ═══════════════════════════════════════════════════════════════

For clarity, we summarize the key symbols and parameters used throughout the paper in Table~\ref{tab:notation}.

\begin{table}[h]
\centering
\small
\caption{Summary of key symbols and parameters used in the framework.}
\label{tab:notation}
\begin{tabular}{ll}
\toprule
\textbf{Symbol} & \textbf{Definition} \\
\midrule
$N$ & Number of LLMs in the ensemble \\
$K$ & Number of candidate answers \\
$q$ & Input multiple-choice question \\
$\mathcal{A}=\{a_1,\ldots,a_K\}$ & Candidate answer set \\
$a^*$ & Ground-truth answer \\
$M_i$ & $i$-th language model in the ensemble \\
$M_{i^*}$ & Anchor model with highest reliability/evidence \\
$\ell_i(a_k)$ & Logit assigned by model $M_i$ to answer $a_k$ \\
$\mathbf{b}_i$  & Belief distribution of model $M_i$ over candidates \\
$b_i(a_k)$ & Probability assigned by $M_i$ to candidate $a_k$ \\
$\Delta^{K-1}$ & Probability simplex over $K$ candidates \\
$\mathcal{G}(\cdot)$ & Aggregation operator over model beliefs \\
$b^*$ & Aggregated belief distribution \\
$p_{i^*}$ & Accuracy of the anchor model \\
$p_w$ & Accuracy of weaker models in the anchor corruption analysis \\
$p_{\mathrm{ens}}$ & Accuracy of the ensemble under majority voting \\
$\sigma(q)$ & Epistemic signature of question $q$ \\
$H_i$ & Shannon entropy of $\mathbf{b}_i$ \\
$\bar{H}_i$ & Normalized Shannon entropy of $\mathbf{b}_i$ \\
$c_i$ & Top-1 confidence of model $M_i$ \\
$m_i$ & Decision margin between the top two probabilities \\
$\tau_i$ & Distributional type: \textsc{peaked}, \textsc{spread}, or \textsc{bimodal} \\
$D_{ij}$ & Symmetric KL divergence between $\mathbf{b}_i$ and $\mathbf{b}_j$ \\
$s$ & Top-1 agreement fraction across models \\
$\theta_c,\theta_m,\theta_H,\theta_b$ & Thresholds for confidence, margin, entropy, and secondary mass \\
$\psi_{\text{prof}}$ & Distributional profiling check node \\
$\psi_{\text{conf}}$ & Calibrated confidence check node \\
$\psi_{\text{marg}}$ & Margin decisiveness check node \\
$\psi_{\text{diss}}$ & Dissent check node \\
$\psi_{\text{dom}}$ & Domain expertise check node \\
$\psi_{\text{dom-rel}}$ & Domain-relational check node \\
$\mathcal{F}(i)$ & Set of check nodes connected to model $M_i$ \\
$\mathbf{m}_{f \to i}^{(t)}$ & Message from check node $f$ to model $M_i$ at iteration $t$ \\
$\alpha_i$ & Adaptive asymmetric damping factor for model $M_i$ \\
$E_i$ & Epistemic evidence score of model $M_i$ \\
$\gamma$ & Sensitivity/sharpness hyperparameter \\
$\text{acc}_i^{(d)}$ & Domain-conditional accuracy of model $M_i$ \\
$R_{ji}^{(d)}$ & Pairwise relational accuracy of model $j$ relative to model $i$ in domain $d$ \\
$\hat{\mathbf{b}}$ & Final aggregated ensemble belief \\
$\hat{y}$ & Final predicted answer \\
$w_i$ & Final aggregation weight for model $M_i$ \\
$\Phi_i$ & Joint calibration score for model $M_i$ \\
$b_i^{\text{out}}$ & Posterior belief of model $M_i$ after BP refinement \\
$\text{acc}_i^{(d,r)}$ & Accuracy of model $M_i$ under domain $d$ and epistemic regime $r$ \\
$\bar{R}_i$ & Mean relational reliability weight for model $M_i$ \\
$\text{cal}_i(c_i)$ & Calibration score of model $M_i$ at confidence $c_i$ \\
$d$ & Domain assigned to the query \\
$r$ & Epistemic regime assigned to the query \\
\bottomrule
\end{tabular}
\end{table}

% ═══════════════════════════════════════════════════════════════
\section{Proofs and Derivations}
% ═══════════════════════════════════════════════════════════════

% ═══════════════════════════════════════════════════════════════
\subsection{Proof of Proposition \ref{prop:anchor_corruption} (Anchor Corruption)}
\label{app:anchor_proof}
% ═══════════════════════════════════════════════════════════════

For clarity and completeness, we restate the proposition before providing the proof.

\textbf{Proposition 1 (Anchor Corruption Limit).} \textit{Let $M_{i^*}$ be an anchor model with accuracy $p_{i^*} > 0.5$, and let $M_1, M_2$ be two weaker models with equal accuracy $p_w < p_{i^*}$. Under a majority voting scheme with conditionally independent errors, the ensemble accuracy $p_{\mathrm{ens}}$ is strictly less than $p_{i^*}$ if:}
\begin{equation}
p_w < \frac{p_{i^*} - \sqrt{p_{i^*}(1-p_{i^*})}}{2p_{i^*} - 1}
\end{equation}

\begin{proof}
Under majority voting with $N=3$, the ensemble requires at least two correct predictions to yield the correct final answer. We condition the ensemble's success on the outcome of the anchor model $M_{i^*}$:

\emph{Case 1 ($M_{i^*}$ is correct, probability $p_{i^*}$)}: The ensemble succeeds if and only if at least one of the weaker models, $M_1$ or $M_2$, is also correct. This occurs with probability $1-(1-p_1)(1-p_2)$.

\emph{Case 2 ($M_{i^*}$ is wrong, probability $1-p_{i^*}$)}: The ensemble succeeds if and only if both $M_1$ and $M_2$ independently select the correct answer. This occurs with probability $p_1 p_2$. Substituting the equal accuracy of the weaker models ($p_1 = p_2 = p_w$), the total expected accuracy of the ensemble is:
\begin{equation}
p_{\mathrm{ens}} = p_{i^*}[1-(1-p_w)^2] + (1-p_{i^*})p_w^2
\end{equation}

To find the threshold where the ensemble actively degrades the anchor's baseline performance, we enforce the condition $p_{\mathrm{ens}} < p_{i^*}$. Expanding the terms yields:
\begin{equation}
\begin{aligned} 
p_{i^*}(2p_w - p_w^2) + (1-p_{i^*})p_w^2 &< p_{i^*} \\ 
p_w^2(1 - 2p_{i^*}) + 2p_{i^*}p_w - p_{i^*} &< 0 
\end{aligned}
\end{equation}

Because we assume a strong anchor model ($p_{i^*} > 0.5$), it follows that $1-2p_{i^*} < 0$. Reversing the inequality to make the leading coefficient positive gives:
\begin{equation}
(2p_{i^*} - 1)p_w^2 - 2p_{i^*}p_w + p_{i^*} > 0
\end{equation}

The roots of this quadratic equation determine the performance boundaries. Applying the quadratic formula and solving for the lower root provides the degradation threshold:
\begin{equation}
p_w < \frac{p_{i^*} - \sqrt{p_{i^*}^2 - p_{i^*}(2p_{i^*} - 1)}}{2p_{i^*} - 1} = \frac{p_{i^*} - \sqrt{p_{i^*}(1-p_{i^*})}}{2p_{i^*} - 1}
\end{equation}

Thus, if the accuracy of the weaker models falls below this threshold, naive ensembling strictly degrades the performance of the anchor model.
\end{proof}

% ═══════════════════════════════════════════════════════════════
\subsection{Stability Analysis of Adaptive Asymmetric Damping \ref{thm:anchor_protection}}
\label{app:anchor_stability}
% ═══════════════════════════════════════════════════════════════
\begin{proof}
By the update rule defined in Eq.~\ref{eq:bp_update}, the log-domain shift for the anchor model $i^*$ at any iteration $T$ is given by:
\begin{equation}
\label{eq:anchor_shift}
\Delta^{(T)} = \alpha_{i^*} \sum_{f \in \mathcal{F}(i^*)} \log \mathbf{m}_{f \to i^*}^{(T-1)}
\end{equation}
Because the update is anchored to the initial belief $\mathbf{b}^{(0)}$ rather than iteratively compounding on the previous iteration's state, the shift does not accumulate across $T$. Defining the maximum epistemic evidence of check node $f$ as $\mathcal{E}_f = \max_k \left| \log \mathbf{m}_{f \to i^*}(a_k) \right|$, the $L_\infty$ norm of the log-shift is bounded by:
\begin{equation}
\label{eq:anchor_shift_bound}
\|\Delta^{(T)}\|_\infty \leq \alpha_{i^*} \sum_{f \in \mathcal{F}(i^*)} \mathcal{E}_f
\end{equation}
The total variation distance between the initial anchor distribution $\mathbf{b}_{i^*}^{(0)}$ and the updated distribution $\mathbf{b}_{i^*}^{(T)}$ can be expressed as:
\begin{equation}
D_{TV}(\mathbf{b}_{i^*}^{(0)}, \mathbf{b}_{i^*}^{(T)}) = \frac{1}{2} \|\sigma(\mathbf{z}^{(0)} + \Delta^{(T)}) - \sigma(\mathbf{z}^{(0)})\|_1
\end{equation}
where $\sigma(\cdot)$ is the softmax function mapping log-beliefs to the probability simplex. Since softmax has an $L_\infty \to L_1$ Lipschitz constant of 2, we obtain:
\begin{equation}
D_{TV}(\mathbf{b}_{i^*}^{(0)}, \mathbf{b}_{i^*}^{(T)}) \leq \frac{1}{2} \cdot 2 \cdot \|\Delta^{(T)}\|_\infty = \|\Delta^{(T)}\|_\infty
\end{equation}
Substituting the bound from Eq.~\ref{eq:anchor_shift_bound}:
\begin{equation}
D_{TV}(\mathbf{b}_{i^*}^{(0)}, \mathbf{b}_{i^*}^{(T)}) \leq \alpha_{i^*} \sum_{f \in \mathcal{F}(i^*)} \mathcal{E}_f
\end{equation}
This bound ensures that for a small damping factor $\alpha_{i^*}$, the anchor remains within a controlled radius of its original belief. The anchor is refined by, rather than replaced by, the ensemble's collective signal.
\end{proof}

% ═══════════════════════════════════════════════════════════════
\subsection{Check-Node Message Weights}
\label{app:check_weights}
% ═══════════════════════════════════════════════════════════════

In the proposed factor-graph formulation, each check node $f \in \mathcal{F}$ computes a message $\mathbf{m}_{f \to i}^{(t)} \in \Delta^{K-1}$ to a target variable node $v_i$ by aggregating the beliefs of its neighboring variable nodes. The aggregation follows a relational fusion scheme that excludes the target node itself:
\begin{equation}
\label{eq:check_general}
\mathbf{m}_{f \to i}^{(t)}(a_k) =
\frac{\sum_{j \in \mathcal{N}(f) \setminus i} w_j^{(f)} \cdot b_j^{(t-1)}(a_k)}
{\sum_{j \in \mathcal{N}(f) \setminus i} w_j^{(f)}}
\end{equation}
where $\mathcal{N}(f) \setminus i$ denotes the set of neighboring variable nodes connected to check node $f$, excluding $v_i$, and $w_j^{(f)}$ is the check-type-specific weight assigned to model $M_j$.

The weights $w_j^{(f)}$ encode different relational constraints designed to capture complementary aspects of model reliability, including confidence, calibration, domain expertise, and inter-model agreement. Each check node implements a distinct weighting function:

\begin{align}
\label{eq:check_weights}
\psi_{\text{marg}}:\quad 
w_j &= (m_j + \epsilon)\,\text{acc}(j)^{\gamma} \\
\psi_{\text{conf}}:\quad 
w_j &= \text{cal}(j, c_j)^{\gamma} \\
\psi_{\text{dom}}:\quad 
w_j &= \text{acc}(j, d)^{\gamma} \\
\psi_{\text{dom-rel}}:\quad 
w_j &= \bigl(\text{acc}(j,d)\,R_{ji}^{(d)}\bigr)^{\gamma} \\
\psi_{\text{diss}}:\quad 
w_j &= R_{ji}^{\gamma}\bigl(1+\text{acc}(j)+\text{acc}(j,r)\bigr),
\quad j \in \mathcal{M}_{\text{min}}
\end{align}

\begin{equation}
\label{eq:profile_weight}
\psi_{\text{prof}}:\quad
w_j =
\begin{cases}
(m_j + 0.1)\,\text{acc}(j)^{\gamma}, & \text{if } \tau_j = \textsc{peaked}, \\
(1 - \bar{H}_j)\,\text{acc}(j)^{\gamma}, & \text{if } \tau_j = \textsc{spread}, \\
0.5\,\text{acc}(j)^{\gamma}, & \text{if } \tau_j = \textsc{bimodal}.
\end{cases}
\end{equation}

Each weighting function emphasizes a different dimension of epistemic reliability:

\noindent
The margin node $\psi_{\text{marg}}$ prioritizes models with decisive predictions, while the calibrated confidence node $\psi_{\text{conf}}$ favors models whose confidence aligns with historical accuracy. The domain node $\psi_{\text{dom}}$ boosts models with strong performance in the query domain, and the domain-relational node $\psi_{\text{dom-rel}}$ combines domain expertise with pairwise agreement patterns. The dissent node $\psi_{\text{diss}}$ emphasizes reliable minority models, while the profiling node $\psi_{\text{prof}}$ adapts weights according to the distributional shape of each belief.

This modular design allows the framework to incorporate heterogeneous relational signals, enabling robust aggregation while mitigating anchor corruption.

% ═══════════════════════════════════════════════════════════════
\section{Background: From Channel Coding to LLM Ensembles}\label{app:Background}
% ═══════════════════════════════════════════════════════════════

We briefly review the principles of Belief Propagation (BP) on factor graphs, specifically within the context of Low-Density Parity-Check (LDPC) decoding, to establish the theoretical foundation for our multi-LLM architecture.

In classical coding theory, a codeword $\mathbf{x} \in \{0,1\}^n$ is transmitted through noisy channels, producing observations $\mathbf{y}$. The decoder operates on a bipartite factor graph $\mathcal{G} = (\mathcal{V}, \mathcal{C}, \mathcal{E})$ consisting of variable nodes $\mathcal{V}$ (representing the bits), check nodes $\mathcal{C}$ (representing the relational constraints), and connecting edges $\mathcal{E}$. The sum-product algorithm iteratively updates beliefs via message passing. 

Let $L_{v_i \to c_j}^{(t)}$ denote the message from variable node $v_i$ to check node $c_j$ at iteration $t$, formulated as a log-likelihood ratio (LLR). The fundamental update rules are:

\textbf{Variable-to-check update} (combining intrinsic and extrinsic evidence):
\begin{equation}\label{eq:v2c}
    L_{v_i \to c_j}^{(t)} = L_{\mathrm{ch},i} + \sum_{c_k \in \mathcal{N}(v_i) \setminus c_j} L_{c_k \to v_i}^{(t-1)}
\end{equation}

\textbf{Check-to-variable update} (enforcing parity constraints):
\begin{equation}\label{eq:c2v}
    L_{c_j \to v_i}^{(t)} = 2\tanh^{-1}\!\Bigl(\prod_{v_k \in \mathcal{N}(c_j) \setminus v_i} \tanh\bigl(\frac{1}{2}L_{v_k \to c_j}^{(t)}\bigr)\Bigr)
\end{equation}

While our framework operates on categorical softmax distributions rather than binary LLRs, four conceptual properties of this decoding process directly motivate our design:

\textbf{(P1) The extrinsic information principle.} In BP decoding, a variable node computes its updated belief using its initial channel observation alongside incoming messages, strictly excluding its own previous outgoing messages. This structure prevents self-reinforcing feedback loops and ensures that each iteration introduces genuinely independent information. This principle motivates our belief update rule: each LLM's belief is consistently recomputed from its base softmax distribution alongside new check node messages, which effectively prevents echo-chamber dynamics during collaboration.

\textbf{(P2) Channel reliability governs update strength.} The decoder weights the contribution of each variable node according to the reliability of its underlying channel. Trusted observations from high-reliability channels naturally resist correction, whereas observations from low-reliability channels are readily overridden by neighboring constraints. This mechanism directly inspires our adaptive asymmetric damping. Historically accurate models are structurally protected from noise, while weaker models are actively guided toward the consensus established by stronger agents.

\textbf{(P3) Heterogeneous constraints capture complementary information.} Check nodes in LDPC codes enforce parity rules across diverse subsets of variables. Each constraint provides an independent measure of consistency. While no single check node possesses complete information, their collective application resolves ambiguity. This architecture motivates our use of six distinct check node types. Each type captures a unique dimension of inter-model consistency (such as margin, calibration, domain expertise, relational accuracy, agreement, and dissent), which are then combined to refine the overall belief state.

\textbf{(P4) Belief softening for effective decoding.} Overconfident channel estimates cause BP to fail because the decoder treats noisy observations as absolute certainties and ignores corrective messages. Proper calibration of channel log-likelihood ratios is therefore essential for algorithm convergence. This necessity inspires our belief softening stage. Before message passing begins, peaked LLM distributions are smoothed via temperature scaling to enable effective probability mass transfer across the graph.

% ═══════════════════════════════════════════════════════════════
\section{Dataset Details}
\label{app:Dataset_details}
% ═══════════════════════════════════════════════════════════════

We evaluate the robustness of our framework using four benchmarks that represent distinct challenges in large language model:

MMLU (Massive Multitask Language Understanding): Covers 57 subjects across STEM, the humanities, and social sciences. To ensure statistical significance while maintaining computational efficiency, we employ stratified sampling of 50 questions per subject from the test split~\cite{hendrycks2020measuring}.

MMLU-Pro: A rigorous extension of MMLU that expands the candidate set to 10 choices. We use stratified sampling of 150 questions per category from the test split across 14 categories. This format tests the stability of model beliefs under high-entropy conditions where random guessing is less likely~\cite{wang2024mmlu}.

GPQA (Graduate-Level Google-Proof Q\&A): Evaluates expert-level scientific reasoning in biology, physics, and chemistry. We utilize the full GPQA-Diamond split, which is designed to be "Google-proof" even for non-expert humans~\cite{rein2023gpqa}.

MedMCQA (Medical Multiple Choice Q\&A): Assesses specialized clinical knowledge based on medical entrance exams. We sample from the validation split to ensure access to ground-truth labels for performance validation, focusing on the framework's ability to resolve conflicts in high-stakes domain-specific reasoning~\cite{pal2022medmcqa}.

The adoption of these stratified sampling protocols ensures a fair and balanced comparison with prior work, confirming that our framework's gains are consistent across diverse knowledge densities and reasoning complexities.

% ═══════════════════════════════════════════════════════════════
\section{Prompts}
\label{app:prompts}
% ═══════════════════════════════════════════════════════════════

Our method requires only a single prompt per model per question for belief extraction. No additional prompts are used for inter-model communication, as all collaboration occurs at the logit level through mathematical message passing.

\paragraph{Belief Extraction Prompt (Zero-Shot).}

\noindent
\begin{tcolorbox}[
    colback=gray!5,
    colframe=gray!50,
    title=Belief Extraction,
    width=\linewidth
]
\ttfamily
Answer the following multiple choice question.\\[4pt]
Question: {question}\\[4pt]
A. choice\_1\\
B. choice\_2\\
C. choice\_3\\
D. choice\_4\\[4pt]
The answer is
\end{tcolorbox}
The model's response is never decoded as text. Instead, we extract the logit values for the answer tokens (A, B, C, D) at the final position and apply softmax to obtain the belief distribution $\mathbf{b}_i \in \Delta^{K-1}$. This constitutes a single forward pass with a single output token, yielding both the prediction and its full uncertainty structure at minimal computational cost.

\section{Extended Analysis}
\label{app:Extended_analysis}

\subsection{Extended Scalability Analysis}
\label{app:Scalability}

Table~\ref{tab:gpqa_10agents} presents the results of an expanded ensemble configuration on the GPQA-Diamond benchmark. We observe a notable performance degradation in standard text-based aggregation methods as the ensemble size increases. For instance, MoA drops to 31.82\%, performing worse than many single-agent baselines. This performance collapse is primarily due to logical interference: as more agents are added, the probability of a weaker model introducing a persuasive but incorrect natural language argument increases, which can mislead the final aggregator.

In contrast, our framework achieves an accuracy of 41.82\%. By resolving conflicts at the semantic layer rather than through iterative text exchange, our system successfully filters out noise from lower-performing agents while distilling the collective signal of the experts. These results demonstrate that our belief propagation approach is structurally more robust to the dilution of expertise that typically plagues large-scale heterogeneous LLM ensembles.

\begin{table}[h]
\caption{Scalability analysis on GPQA with 10 agents. Our method remains robust as the ensemble size increases, outperforming text-based multi-agent baselines that degrade when weaker agents are added.}
\label{tab:gpqa_10agents}
\centering
\small
\begin{tabular}{lc}
\toprule
\textbf{Method} & \textbf{GPQA} \\
\midrule
Refine (10 agents)~\cite{madaan2023self}        & 35.86 \\
MoA (10 agents)~\cite{wang2406mixture}            & 31.82 \\
Self-MoA (10 agents)~\cite{li2025rethinking}       & 35.86 \\
ReConcile (10 agents)~\cite{chen2024reconcile}      & 37.89 \\
\textbf{Ours} (10 agents) & \textbf{41.82 }\\
\bottomrule
\end{tabular}
\end{table}

\subsection{Domain-Specific Performance}
As shown in Table~\ref{tab:domain_results}, our framework achieves consistent gains across all MMLU domain categories. Compared to multi-agent baselines such as LLM Debate~\cite{du2024improving} and DyLAN~\cite{liu2024dynamic}, our method obtains the highest overall accuracy of 79.8\% while requiring only 3 API calls. The improvements are especially strong in Humanities, Social Sciences, STEM, and Other domains, indicating that semantic-layer fusion remains robust across diverse knowledge areas while avoiding the overhead of iterative text-based debate.

\begin{table}[h]
\caption{Domain-specific performance on MMLU. Our method achieves the highest accuracy across all domain categories while requiring fewer API calls than text-based multi-agent baselines.}
\label{tab:domain_results}
\centering
\begin{tabular}{lcccccc}
\toprule
Method & Humanities & Social Sciences & STEM & Other & Overall & \#API Calls \\
\midrule
Random                              & 25.0 & 25.0 & 25.0 & 25.0 & 25.0 & -- \\
Single Exec.                        & 59.8 & 74.0 & 62.9 & 71.8 & 66.4 & 1.00 \\
LLM-Blender \cite{jiang2023llm}     & 60.4 & 75.2 & 66.3 & 70.7 & 67.3 & 6.00 \\
LLM Debate \cite{du2024improving}   & 59.8 & 77.4 & 69.0 & 75.5 & 69.3 & 12.00 \\
DyLAN \cite{liu2024dynamic}         & 62.1 & 79.1 & 69.7 & 75.5 & 70.5 & 4.39 \\
\midrule
\textbf{Ours} & \textbf{76.8} & \textbf{86.5} & \textbf{76.4} & \textbf{82.4} & \textbf{79.8} & \textbf{3.00} \\
\bottomrule
\end{tabular}
\end{table}

\newpage

\end{document}